\documentclass[runningheads]{llncs}

\usepackage{graphicx}
\usepackage{amsfonts}
\usepackage{bm}
\usepackage{amsmath}
\usepackage{amssymb}
\usepackage{xspace}
\usepackage{xcolor}
\usepackage{nicefrac}

\newcommand{\uset}{{\mathcal U}}
\newcommand{\iset}{{\mathcal I}}
\newcommand{\jset}{{\mathcal J}}
\newcommand{\lset}{{\mathcal L}}
\newcommand{\ssset}{{\mathcal S}}
\newcommand{\aset}{{\mathcal A}}
\newcommand{\bset}{{\mathcal B}}
\newcommand{\cset}{{\mathcal C}}

\newcommand{\wone}{\textsf{W[1]}\xspace}

\newcommand{\wonehard}{{{\wone}-hard}\xspace}
\newcommand{\fpt}{\textsf{FPT}\xspace}
\newcommand{\poly}{\textsf{poly}\xspace}

\newcommand{\ksp}{$k$\mbox{\sf -Set Packing}\xspace}

\newcommand{\spp}{\mbox{\sc Set Packing}\xspace}

\newcommand{\ecp}{\mbox{\sc Exact Covering}\xspace}
\newcommand{\pspfull}{{\sc Parameterized Set Packing}\xspace}
\newcommand{\psp}{\mbox{\sc PSP}\xspace}
\newcommand{\cpsp}{\mbox{\sc Compact \psp}\xspace}

\newcommand{\oneinthree}{\mbox{\sc $1$-in-$3$-SAT}\xspace}
\newcommand{\threesat}{\mbox{\sc $3$-SAT}\xspace}

\newcommand{\kclq}{$k$\mbox{\sc -Clique}\xspace}
\newcommand{\rvsm}{$r$\mbox{\sc -VectorSum}\xspace}
\newcommand{\crvs}{\mbox{\sc Compact $r$-VectorSum}\xspace}

\newcommand{\sgifull}{\mbox{\sc Subgraph Isomorphism}\xspace}
\newcommand{\sgi}{\mbox{\sc SGI}\xspace}

\newcommand{\clq}{\mbox{\sc Clique}\xspace}

\newcommand{\ignore}[1]{}

\newcommand{\cmp}[1]{{\sf comp}#1}

\newcommand{\scr}{{\sc Set $r$-Covering}\xspace}
\newcommand{\ecr}{{\sc Exact $r$-Covering}\xspace}
\newcommand{\cscp}{{\sc Compact \scr}\xspace}
\newcommand{\cecp}{{\sc Compact \ecr}\xspace}

\title{On the Parameterized Complexity of Compact Set Packing}
\author{Ameet Gadekar}
\institute{Aalto University, Espoo, Finland \email{ameet.gadekar@aalto.fi}}

\begin{document}
	\maketitle              
	\begin{abstract}
		The \spp problem is, given a collection of sets $\ssset$ over a ground set $\uset$, to find a maximum collection of sets that are pairwise disjoint. 
		The problem is among the most fundamental NP-hard optimization problems that have been studied extensively in various computational regimes.
		The focus of this work is on parameterized complexity, \pspfull (\psp): Given $r \in {\mathbb N}$, is there a collection $ \ssset' \subseteq \ssset: |\ssset'| = r$ such that the sets in $\ssset'$ are pairwise disjoint?
		Unfortunately, the problem is not fixed parameter tractable unless $\wone = \fpt$, and, in fact, an ``enumerative'' running time of $|\ssset|^{\Omega(r)}$ is required unless the exponential time hypothesis (ETH) fails.
		This paper is a quest for  tractable instances of \spp from parameterized complexity perspectives.
		We say that the input $(\uset,\ssset)$ is ``compact'' if $|\uset|  = f(r)\cdot\Theta(\poly( \log |\ssset|))$, for some $f(r) \ge r$. In the \cpsp problem, we are given a compact instance of \psp.
		In this direction, we present a ``dichotomy'' result of \psp: When  $|\uset| = f(r)\cdot o(\log |\ssset|)$, \psp is in \fpt, while for $|\uset| = r\cdot\Theta(\log (|\ssset|))$, the problem is \wonehard; moreover, assuming ETH,  \cpsp does not  admit $|\ssset|^{o(r/\log r)}$ time algorithm even when $|\uset| = r\cdot\Theta(\log (|\ssset|))$. 
		Although certain results in the literature imply hardness of compact versions of related problems such as \scr and \ecr, these constructions fail to extend to \cpsp. A novel contribution of our work is the identification and construction of a gadget, which we call Compatible Intersecting Set System pair, that is crucial in obtaining the hardness result for \cpsp.
		Finally, our framework can be extended to obtain improved running time lower bounds for  \crvs.
	\end{abstract}
	
	%
	%
	%
	
	\section{Introduction}
	Given a graph $G = (V, E)$, the problem of finding a maximum-size subset of disjoint edges (matching) is tractable, but its generalization to hypergraphs, even when the edge length is 3, is \textsf{NP}-hard.
    This general problem is known as the \textsf{Hypergraph Matching} problem. The hyper-graph $H=(W,F)$ can be equivalently viewed as a set system $(\uset,\ssset)$,
	where the universe (or the ground set) $\uset$ corresponds to the vertex set $W$ and $\ssset$ corresponds to the collection of hyperedges $F$. Then finding a maximum matching in $H$ is equivalent to finding maximum number of pairwise disjoint sets (packing) in $\ssset$. Hence the \textsf{Hypergraph Matching} problem is also known as the \spp problem, which is a fundamental problem in combinatorial optimization with numerous applications.
 While this problem captures many classical combinatorial problems such as maximum independent set (or maximum clique), $k$-dimensional matching and also, some graph packing problems~\cite{chataigner2009approximation,hassin2006approximation}, this generalization also makes it intractable in several regimes. 
	One computational regime in which \spp has been explored extenstively is approximation algorithms. 
	Since \spp generalizes the maximum independent set problem~\cite{ausiello1980structure}, it inherits the inapproximability of the latter problem~\cite{10.5555/874062.875499}. This immediately implies that the trivial approximation of picking simply one set in the packing is roughly the best to hope for. Furthermore, approximations in terms of $|\uset|$ are also not hopeful since the result also implies inapproximability bound of $|\uset|^{1/2 -\epsilon}$, which is matched by~\cite{halldorsson2000independent}. To combat these intractabilities, various restrictions of \spp have been studied. Particularly, a restriction where the size of the sets in $\ssset$ is bounded by some integer $k$,  which is known as \ksp, is also a well-studied problem. However, \ksp captures the independent set problem in bounded degree graphs, which again is a notoriously hard problem to approximate beyond the ``trivial'' bound~\cite{austrin2009inapproximability,bansal2018lovasz}. While \cite{hazan2006complexity} improves the lower bound for \ksp to $\Omega(k/\ln k)$, the best known approximation is $(k+1+\epsilon)/3$ \cite{chan2012linear,cygan2013improved}, yielding a logarithmic gap between the bounds. 
	Besides approximation algorithms, \spp has also been studied from  the parameterized complexity perspectives (with the standard parameter on the size of an optimal packing solution). 
	In this problem, known as the \pspfull  (\psp) problem, we are given an instance $(\uset,\ssset,r)$ and the task is to decide if there exists a packing of size $r$. Unfortunately, even \psp remains intractable and is, actually, \wone-complete \cite{Downey:2012:PC:2464827}. In fact, \textit{Exponential Time Hypothesis} (ETH) implies that the trivial \textit{enumerative algorithm} running in $O^{*}(|\ssset|^r)$ time  to find an $r$-packing is asymptotically our best hope~\cite{10.5555/2815661}. The algorithmic outlook for \psp worsens further due to \cite{chalermsook2017gap}, which rules out $o(r)$-\fpt-approximation algorithm assuming the \textit{Gap Exponential Time Hypothesis} (Gap-ETH). 
	Assuming a weaker hypothesis of $\fpt \neq \wone$, very recently~\cite{s2021polynomial} showed that there is   no \fpt algorithm  for \psp problem that finds a packing of size $r/r^{1/H(r)}$, for any increasing function $H(\cdot)$, when given a promise that there is an $r$-packing in the instance.
	Thus, the flurry of these negative results make it likely that \spp problem is intractable in  all computational regimes.
	
	In this paper, 
	we consider \psp on \textit{compact} instances. We say that an instance $(\uset, \ssset, r)$ of $\psp$ is compact if  $|\uset| = f(r) \cdot \Theta( \poly(\log |\ssset|))$, for some function $f(r) \ge r$, that is, the universe is relatively small compared to the number of sets\footnote{In fact there is another way to define compactness:  when $|\ssset| = f(r) \cdot\Theta( \poly(\log |\uset|))$. However in this case, the \textit{enumerative algorithm} running in time $O^{*}(|\ssset|^r)$ is already fixed parameter tractable~\cite{cai2001subexponential}. Thus, the interesting case is when the universe is compact, which is the case we will be focusing on.}.    
	Besides the algorithmic motivation, compact instances have recently been used as an ``intermediate step'' to prove FPT inapproximability results of the (non-compact) classical  problems (see, e.g., \cite{bhattacharyya_et_al:LIPIcs:2016:6362,lin:LIPIcs:2019:10657} where the compact instances were used in proving FPT-inapproximability of the $k$-EvenSet and Dominating Set). 
	We hope that studying \cpsp would lead to some ideas that would be useful in proving tight FPT inapproximability of \psp (that is, to weaken the Gap-ETH assumption used in~\cite{chalermsook2017gap}). 
	
	%
	\subsection{Our Results}
	Our main result is the following dichotomy of \pspfull. 
	\begin{theorem}[Dichotomy] \label{thm:dich}
		The following dichotomy holds for \psp.
		\begin{itemize}
			\item[$\circ$] If $|\uset| = f(r)\cdot o(\log |\ssset|)$, for any $f$, then \psp is in \fpt. 
			\item[$\circ$] \psp remains \wonehard even when $|\uset| = r \cdot \Theta(\log |\ssset|)$.
		\end{itemize}
	\end{theorem}
	The algorithmic result follows from well-known dynamic programming based algorithms~\cite{10.1137/070683933,10.5555/2815661} that run in time $O^*(2^{|\uset|})$, and observing that this running time is fixed parameter tractable~\cite{cai2001subexponential} when $|\uset| = f(r)o(\log |\ssset|)$.
	However, for completeness, in Appendix~\ref{ss:dialgo} we present a simple dynamic program algorithm based on finding longest path in a DAG  running in time  $O^*(2^{|\uset|})$.
	The main contribution of our work is the \wone-hardness of \psp even when $|\uset| = r \cdot \Theta(\log |\ssset|)$. Towards this, we show an \fpt-reduction from  \sgifull (\sgi)  to \cpsp. The hardness result follows since \sgi is \wonehard. In fact, our hardness result can be strengthened assuming Exponential Time Hypothesis (ETH)~\cite{10.5555/2815661} to obtain the following result. 
	\begin{theorem}\label{thm:enuhard}
		\cpsp requires time $|\ssset|^{\Omega \left(r/\log r\right)}$ even when $|\uset| = r \cdot \Theta(\log |\ssset|)$, unless ETH fails.
	\end{theorem}
	The result of Theorem~\ref{thm:enuhard} follows from the ETH-hardness result of \sgi due to \cite{marx2007can}, and from the fact that	the hardness reduction of Theorem~\ref{thm:dich} is parameter preserving up to a multiplicative constant.
	Note that since \psp can be trivially solved by enumeration in  time $O^*\left(|\ssset|^r\right)$, the above result says that, even for the compact instances this is essentially our best hope, up to a $\log$ factor in the exponent.
	An interesting consequence of the dichotomy theorem coupled with Theorem~\ref{thm:enuhard} is the fact that, as soon as instances get asymptotically smaller, not only we beat the enumerative algorithm, but we  actually obtain an \fpt algorithm.
	We would like to remark that the universe size in Theorem~\ref{thm:enuhard} is tight (up-to $\log r$ factor) since having $|\uset| = o(r/\log r) \cdot \Theta(\log |\ssset|)$ would already allow $|\ssset|^{o(r/\log r)}$ time algorithm.
    Further, note that for \wone-hardness, it is sufficient to have $|\uset| = f(r) \Theta(\log |\ssset|)$, for some $f$, since we can  add $f(r)-r$ new sets each with a unique dummy element and inflate the parameter to $f(r)$. However, this is not true for ETH based running time lower bounds as such inflation fail to transfer the lower bounds asymptotically.	 

    We would like to remark that, after our article was made public, Huairui Chu~\cite{https://doi.org/10.48550/arxiv.2302.09220} improved the lower bound of Theorem~\ref{thm:enuhard} to rule out $|\ssset|^{\Omega(r)}$ time.

	Finally, we extend our construction framework (Theorem~\ref{thm:ckvs}) to improve the running time lower bound (matching the trivial upper bound up to a $\log$ factor in the exponent) for the compact version of \rvsm: Given a collection $\cset$ of $N$ vectors in $\mathbb{F}^d_2$, and a target vector $\vec{b} \in \mathbb{F}^d_2$, \rvsm asks if there are $r$ vectors in $\cset$ that sum to $\vec{b}$. \crvs is defined when $d = f(r) \cdot \Theta(\poly(\log N))$, for some $f(r)\ge r$.
	\begin{theorem}\label{thm:ckvs}
		\crvs requires time $N^{\Omega(r/\log r)}$, even when $d=r \cdot \Theta(\log N)$, unless ETH fails.
	\end{theorem}
    The present bound of \cite{bhattacharyya_et_al:LIPIcs:2016:6362} rules out $N^{o(\sqrt{r})}$ time under ETH. 
    The proof of this theorem is present in Appendix~\ref{sec:ckvs}.


	\subsection{Our contributions and comparison to existing works} \label{ss:comp}
	In this section, we compare our contribution with existing works to highlight its significance. 
	To our best knowledge, the compact version of combinatorial problems has not previously been formalized and investigated. However, several existing reductions already imply the hardness of compact version of some of the combinatorial problems. Here we review and compare the related results.
	
	\paragraph{\textbf{Our Contribution.}} 
	As far as we know, there are no results showing \wone-hardness of \cpsp, and hence the corresponding dichotomy (Theorem~\ref{thm:dich}). The key contribution of this paper is to show the hardness result for \cpsp.
	On the way, we also show an ETH-based almost tight running time lower bound for \cpsp, with tight (up-to $\log r$ factor) universe size  $|\uset|= r\cdot \Theta(\log |\ssset|)$.
	Interestingly, we show both of these results with a single \fpt reduction. In addition, we extend our framework to improve the running time lower bounds for \crvs.
	
    Next we survey some known hardness results for \scr in the compact regime and argue their limitations in extending them to \psp. In particular, \textbf{\cite[Lemma~25]{10.1007/978-3-642-40450-4_57}}  shows a reduction from \sgi to  a variant of \scr called \cecp, where we want to find an $r$-packing that is also a covering (in fact they show hardness for \cscp. But a closer inspection of their construction shows that the intended set cover is also a packing).
	The high level idea of the construction is similar to ours:  first assign each vertex of $G$ a logarithm length binary pattern vector. Then, create two kinds of sets: $V$-sets that capture the mapping of the vertices and $E$-sets that capture the mapping of edges. The idea is to use the pattern vectors to create these sets so that there is an isomorphic copy of $H$ in $G$  if and only if there are $|V_H|$ many $V$-sets and $|E_H|$ many $E$-sets covering the universe exactly once.
    However, if we consider the Soundness (No case) proof of this reduction, then  it crucially relies on the fact that no candidate solution can cover the entire universe exactly once.  In fact, it is quite easy to find $r$ sets that are mutually disjoint but do not form a cover.
	Therefore, it fails to yield hardness for \cpsp. The heart of our construction lies in ensuring that in the No case, any $r$ sets intersect. To this end, we construct a combinatorial gadget called  Compatible Intersecting Set System (ISS) pair. This gadget is a pair of set systems $(\aset,\bset)$ over a universe $U$ that guarantees two properties: First, every pair of sets within each set system intersects, and second, for any set $a \in \aset$, there exists $b\in \bset$ such that $a$  intersects every set in $\bset$ except $b$. Further, we present a simple greedy algorithm that finds such compatible ISS pair $(\aset,\bset)$ over a universe of size $N$, each having roughly $2^{\Omega(N)}$ sets. Note that this gadget, which we use to build our compact hard instance, also has a ``compact" universe.
    While, on the other hand, \textbf{\cite{10.1145/3188745.3188896}}	shows  \cscp is \wonehard using a reduction from \kclq to \scr with $r=\Theta(k^2)$ and $|\uset| = r^{3/2} \cdot \Theta(\log |\ssset|)$, but does not yield a tight ETH- based running time lower bound. 
	In contrast,	\cite{puatracscu2010possibility} shows such tight ETH lower bound for \cscp: requiring time $|\ssset|^{\Omega\left(r\right)}$, which can be easily modified to obtain a similar running time lower bound for \cecp	(by reducing from  \oneinthree, instead from \threesat).

	\subsection{Overview of Techniques} \label{ss:tech}
	In this section we sketch the main ideas of our hardness proof of Theorem~\ref{thm:dich}. To this end, we present a reduction from \sgi, which asks, given a graph $G$ on $n$ vertices and another graph $H$ with $k$ edges, if there is a subgraph of $G$ isomorphic (not necessarily induced)  to $H$, with parameter $k$. The reduction produces an instance $\iset = (\uset,\ssset,r)$ of \psp  in \fpt time  such that $r = \Theta(k)$ and $|\uset| = \Theta(r \log |\ssset|)$.   
	We remark that the classical reduction given in \cite{Downey:2012:PC:2464827} also has parameter $r = \Theta(k)$, but $|\uset|$ is linear in the size of $G$, which is the size of $|\ssset|$. Below we present a reduction that constructs a compact instance, but falls short in achieving its goal. However, it illustrates some of the main ideas that form the basis of the actual hardness proof. This failed attempt also highlights the crucial properties of the gadget that are necessary for the correct reduction.
	
	Our reduction constructs the instance $\iset = (\uset,\ssset,r)$ of \psp using a special set system gadget -- which we call the Intersecting Set System (ISS) gadget. A set system $\aset=(U_A,S_A)$ with $M$ sets over $N$ elements is called an $(M,N)$-Intersecting set system, if every pair $s^i, s^j \in S_A$ intersects (i.e., $s^i, s^j$ has a non-empty intersection).
	We show how to efficiently construct an $(M,N)$-ISS $\aset=(U_A,S_A)$ with $M = 2^{N-1}$.  Let $U_A = \{1,2,\cdots,N+1\}$. Then, for every subset $s\subseteq \{2,3,\cdots, N+1\}$, add the set $s' := \{1\} \cup s$ to $S_A$. Note that $\aset$ has a compact universe since $|U_A| = \log_2 M + 1 = \log_2 |S_A|+1$, which is  crucial in  constructing a compact instance of \psp.
	We are now ready to present the reduction using this compact ISS gadget.
	Let the given instance  of \sgi be $\jset = (G=(V_G,E_G), H=(V_H,E_H),k)$. Let $\ell := |V_H|$, $n:= |V_G|$ and $m:= |E_G|$. Further, let $V(G) = \{1,\cdots,n\}$.  Note that $\ell \le 2k$, since isomorphic sub-graph in $G$ to $H$ is not necessarily induced. 
	Let $\aset=(U_A,S_A)$ be the $(M,N)$-ISS gadget specified above with  $N=\lceil \log n \rceil +1$.
	Since $M \ge n$, assume $S_A=\{s^\alpha\}_{\alpha \in V(G)}$ by arbitrarily labeling sets in $S_A$ and ignoring the sets $s^\alpha, \alpha>n$.
	We construct an instance $\iset = (\uset,\ssset,r)$ of \cpsp as follows. 
	For every $v \in V_H$, and $w \in N_H(v)$, let $\aset_{v,w}=(U_{v,w}, S_{v,w})$ be a distinct copy of ISS $\aset$ (that is, the universes $\{U_{v,w}\}_{w \in N_H(v)}$ are disjoint) with the same labeling of sets as that of $S_A$. Note that $\aset_{v,w}$ and $\aset_{w,v}$ are distinct copies of $\aset$.
	Let $U_v:= \cup_{w \in N(v)} U_{v,w}$. The universe $\uset$ in $\iset$ is defined as $
	\uset = \cup_{v\in V_H} U_v.
	$
	Now for $\ssset$, we will construct two types of sets, that we call $V$-sets and $E$-sets. 
 	For every $\alpha \in V(G)$ and $v \in V_H$, add the set $S_{\alpha \mapsto v} := \cup_{w \in N(v)} s^\alpha_{v,w}$ to $\ssset$.
	These sets are referred  as $V$-sets. 
	For each edge $(\alpha,\beta) \in E_G$ and each edge $(v,w) \in E_H$, add the set $S_{(\alpha,\beta) \mapsto (v,w)} := \bar{s}^\alpha_{v,w} \cup \bar{s}^\beta_{w,v}$ to $\ssset$. 
	These sets are called $E$-sets. 
	Finally, setting the parameter $r=\ell + k$, concludes the construction of \psp instance $\iset = (\uset,\ssset,r)$. First, note that for the base ISS gadget $\aset = (U_A,S_A)$, we have that $|U_A| = \Theta(\log n)$. Hence, $|\uset| = \sum_{i \in [\ell]} \sum_{j \in [d(v_i)]} |U_A| = \Theta(k \log n)$, where as $|\ssset| = \Theta(mk+n\ell) = \Theta(n^2k)$. Since $r =\Theta(k)$, we have  $|\uset| = \Theta(r \log |\ssset|)$, yielding a \cpsp instance.

	To illustrate the main ideas, we analyze the completeness and discuss how the soundness fails. 
	In the completeness case, let us assume that there exists an injection $\phi: V_H \rightarrow V_G$ which specifies the isomorphic subgraph in $G$. 
    Consider $T \subseteq \ssset$ as $T = \{S_{\phi(v) \mapsto v}\}_{v \in V_H} \bigcup \{S_{(\phi(v), \phi(w)) \mapsto (v,w)} \}_{(v,w) \in E_H}$.
	Notice that we have chosen $\ell+k$ sets from $\ssset$. 
    To see that $T$ forms a packing, fix $v,w \in V_H$ such that $w \in N(v)$.
	Let $T\mid_{U_{v,w}}$ be the restriction of $T$ on $U_{v,w}$ (Formally, $T\mid_{U_{v,w}} := \{t \cap U_{v,w} : t \in T\}$).
	Then note that $T\mid_{U_{v,w}}$ forms a packing since
	$
	T\mid_{U_{v,w}} = \{ s^{\phi(v)}_{v,w},\  \bar{s}^{\phi(v)}_{v,w} \}.
	$
	
	For soundness, we show the proof for a simpler case. Suppose $T \subseteq \ssset$ with $|T|= r$ is a packing with at most one $V$-set from each vertex of $G$. Further, assume $|T_V| = \ell$ and $|T_E| = k$, where $T_V$ and $T_E$ denote the $V$-sets and $E$-sets of $T$ respectively. Finally, we also assume that $T$ covers $\uset$.
 	Let $V_H = \{v_1,\cdots,v_\ell\}$. Relabel the sets in $T_V$ as  $T_V=\{T^i_V : \exists S_{\alpha\mapsto v_i} \in T_V, \text{ for some } \alpha \in V_G, v_i \in V_H \}$.
	Now consider $V' := \{\alpha \mid T_V^i = S_{\alpha \mapsto v_i}\} \subseteq V_G$, and relabel the vertices of $V'$ as
	$V' = \{\alpha'_1, \cdots, \alpha'_\ell\}$, where $\alpha'_i = \alpha$ such that $T_V^i = S_{\alpha \mapsto v_i}$. We claim that $G[V']$ is isomorphic to $H$ with injection $\phi:V_H \rightarrow V_G$ such that $\phi(v_i) = \alpha'_i$. To this end, we show  $(v_i,v_j) \in E_H \implies (\phi(v_i),\phi(v_j)) \in E_{G[V']}$. Note that since $\aset_{v_i,v_j}$ is an ISS, $T_V\mid_{U_{v_i,v_j}} = T^i_V\mid_{U_{v_i,v_j}} = s^{\alpha'_i}_{v_i,v_j}$. Further, combining the fact that $T$ is a packing covering $\uset$ with the fact $|T_E|=k$, we have that  $T_E\mid_{U_{v_i,v_j}} =  \bar{s}^{{\alpha}'_i}_{v_i,v_j}$.
	Similarly, it holds that $T_V\mid_{U_{v_j,v_i}} = T^j_V\mid_{U_{v_j,v_i}} = s^{\alpha'_j}_{v_j,v_i}$, and hence $T_E\mid_{U_{v_j,v_i}} =  \bar{s}^{{\alpha}'_j}_{v_j,v_i}$. But this implies that $S_{(\alpha'_i,\alpha'_j) \mapsto (v_i,v_j)} \in T_E$, which means  $(\phi(v_i),\phi(v_j)) = (\alpha'_i,\alpha'_j)  \in E_{G[V']}$, as desired.
    However, for the general case, we would require a gadget that enforces all the above assumptions in any candidate packing.
		
		

		\subsection{Open problems.}
	 An interesting direction is \fpt approximating \cpsp: Given a promise that there is an $r$-packing, is it possible to find a packing of size $\omega(1)$ in \fpt time? Note that for the general \psp problem, there is no $o(r)$ \fpt-approximation, assuming Gap-ETH. However,  recent results~\cite{DBLP:conf/stoc/Lin21,s2021polynomial} use a weaker assumption of $\wone \ne \fpt$ but also obtain weaker \fpt-inapproximibility.
  It is also interesting to show such hardness of approximation  for \cpsp.

	\section{Preliminaries} \label{sec:prel}
	
	\subsection{Notations}
 For $q \in \mathbb{N}$, denote by $[q]$, the set $\{1,\cdots, q\}$. For a finite set $[q]$ and $i \in [q]$, we overload $`+'$ operator and denote by $i+1$ as the (cyclic) successor of $i$ in $[q]$. Thus, the successor of $q$ is $1$ in $[q]$. 
	All the $\log$s are in base $2$.
	For a graph $G=(V,E)$ and a vertex $v \in V$, denote by $N(v)$, the set of vertices adjacent to $v$. Further, $d(v)$ denotes the degree of $v$, i.e., $d(v) := |N(v)|$. 
	For a finite universe $U$ and $s \subseteq U$, denote by $\bar{s}$ as the complement of $s$ under $U$, i.e., $\bar{s} := U \setminus s$. Similarly, for a family of sets $S=\{s_1,\cdots,s_M\}$ over $U$, we denote by $\cmp(S)=\{\bar{s}_1,\cdots,\bar{s}_M\}$. 
	Further, for a subset $s \subseteq U$ and a sub-universe $U' \subseteq U$, denote by $s \mid_{U'}$ as the restriction of $s$ on sub-universe $U'$, i.e., $s \mid_{U'} := s \cap U'$. Similarly, for a family of sets $S=\{s_1,\cdots,s_M\}$ over $U$, denote by $S\mid_{U'}$ as the restriction of every set of $S$ on $U'$, i.e., $S\mid_{U'}:= \{s_1\mid_{U'},\cdots,s_M\mid_{U'}\}$. For a set system $A = (U_A,S_A)$, we denote the complement set system by $\bar{A} = (U_A,\cmp(S_A))$. For $s,t \subseteq U$,  we say $s$ and $t$ intersects if $s \cap t \neq \emptyset$.
	
	

	\subsection{Parameterized Complexity} \label{app:pc}
	The parameterized complexity theory concerns computational aspects of languages $(L,\kappa)$ over a fixed and finite alphabet $\Sigma$, where $L \subseteq \Sigma^*$, and $\kappa : \Sigma^* \rightarrow \mathbb{N}$, called the parameter, is a polynomial time computable function. 
	Thus a parameterized problem is a classical problem together with a parameter $\kappa$. As an example consider the following  classical NP-complete problem.
	
	\noindent\fbox{%
		\parbox{\textwidth}{%
			\clq\\
			\textit{Instance: } A graph $G$ and  $k \in \mathbb{N}$\\
			\textit{Problem: } Decide if $G$ has a clique of size $k$
		}%
	}
	\\
	
	Now consider a parameterized version of \clq defined by $\kappa(G,k) := k$.
	
	\noindent\fbox{%
		\parbox{\textwidth}{%
			\kclq\\
			\textit{Instance: } A graph $G$ and  $k \in \mathbb{N}$\\
			\textit{Parameter: } $k$\\
			\textit{Problem: } Decide if $G$ has a clique of size $k$
		}%
	}
	When the parameter $\kappa$ represents the size of  solution, then it is called \textit{natural} parameter. 
	\begin{definition} [Fixed Parameter Tractable]
		A parameterized problem $(L,\kappa)$ is called \textit{fixed parameter tractable} if there is an algorithm $A$, a constant $c$ and a computable function $f: \mathbb{N} \rightarrow \mathbb{N}$ such that on all inputs $y=(x,k)$, $A$ decides whether $x$ is in $L$  and runs in time at most $f(k)\cdot |x|^c$.
	\end{definition}
	The complexity class \fpt is the set of all fixed parameter tractable problems. 
	In this paper, we consider parameterized problems with natural parameter i.e. $\kappa$ represents the size of solution.
	Once we define the class \fpt, the next natural thing is to define \textit{parameterized reduction} or \textit{\fpt-reduction} with the intention that such a reduction from parameterized problem $Q$ to another parameterized problem $Q'$ allows converting an \fpt algorithm of $Q'$ to an \fpt algorithm of $Q$.
	
	\begin{definition} [\fpt-reduction]
		An \fpt-reduction from $Q \subseteq \Sigma^* \times \mathbb{N}$ to $Q' \subseteq \Sigma^* \times \mathbb{N}$ is an algorithm $R$ mapping from $\Sigma^* \times \mathbb{N} $ to $\Sigma^* \times \mathbb{N}$ such that for all $y=(x,k), R(y) \in Q'$ if and only if $y \in Q$, and for some computable function $f$ and a constant $c$, $R(y)$ runs in time $f(k)\cdot |x|^c$ and $R(y)=(x',k')$, where $k'\le g(k)$ for some computable function $g$.
	\end{definition}
	An extensive treatment of the subject can be found in \cite{10.5555/2815661,Downey:2012:PC:2464827,Flum:2006:PCT:1121738}.

	\subsection{Problem definitions} 
	\begin{definition}[\pspfull (\psp)]
		Given a collection of sets $\ssset =\{S_1,\dots, S_m\}$ over an universe $\uset = \{e_1,\dots,e_n\}$ and an integer $r$, the \psp problem asks if there is a collection of sets $\mathcal{S}' \subseteq \mathcal{S}$ such that $|\mathcal{S}'| = r$ and, $S_i \cap S_j = \emptyset$ for every $S_i \ne S_j \in \mathcal{S'}$.
		An instance of \psp  is denoted as $ (\uset,\ssset,r)$.
	\end{definition}
 \cpsp is defined when the instances have $|\uset|= f(r) \cdot \Theta(\poly(\log |\ssset|))$, for some function $f(r)\ge r$.
 
	Given two graphs $G=(V_G,E_G)$ and $H=(V_H,E_H)$, a homomorphism from $H$ to $G$ is a map $\phi: V_H \rightarrow V_G$ such that if $(v_i,v_j) \in E_H$ then $(\phi(v_i), \phi(v_j)) \in E_G$.
	\begin{definition}[\sgifull (\sgi)]
		Given a graph $G=(V_G,E_G)$ and a smaller graph $H=(V_H,E_H)$ with $|E_H|= k$, the \sgi problem asks if	there is an injective homomorphism  from $H$ to $G$.
		An instance of \sgi is denoted as $(G=(V_G,E_G), H=(V_H,E_H),k)$.
	\end{definition}
	The parameterized version of \sgi has parameter $\kappa = |E_H| = k$.
	Without loss of generality, we  assume 
 $|V_H| \le 2k$, and every vertex of $H$ has degree at most $k$.

\section{Dichotomy of \psp} \label{sec:mainpf}
	In this section we prove the hardness part of the dichotomy theorem (Theorem~\ref{thm:dich}). First in Section~\ref{ss:iss}, we identify the gadget and its associated properties that are crucial for the reduction.
	Then, in Section~\ref{ss:dihardness}, using this gadget, we show an  \fpt-reduction from \sgi to \cpsp. 

	\subsection{Compatible Intersecting Set System Pair} \label{ss:iss}
	A set system $\aset = (U_A,S_A)$ is called an $(M,N)$-Intersecting set system (ISS), if it contains $M$ sets over $N$ elements such that  every pair $s,t \in S_A$ intersects.
	\begin{definition} [Compatible ISS pair]
		Given two ISS $\aset=(U,S_A)$ and $\bset=(U,S_B)$ on a universe $U$, we say that $(\aset,\bset)$ is a compatible ISS pair if there exists an efficiently computible bijection $f: S_A \rightarrow S_B$ such that
		\begin{itemize}
			\item(Complement partition) $\forall s \in S_A$, $s$ and $f(s)$ forms a partition of $U$, and
			\item(Complement exchange)   $\forall s \in S_A$,  $\aset_s :=(U,(S_A \setminus \{s\}) \cup \{f(s)\})$ is an ISS.
		\end{itemize}
	\end{definition}
	Since $f$ is as bijection, we have $|S_A| = |S_B|$, and $\forall t \in S_B$, the set system $\bset_{t} :=(U,(S_B \setminus \{t\}) \cup \{f^{-1}(t)\})$ is also an ISS. Also, for $(s,t) \in (S_A,S_B)$ if $s \cup t = U$, then $t=f(s)$.
	%
	%
    The following lemma efficiently constructs  a compatible $(M,N)$-ISS pair, which is a key ingredient in our hardness proof. 
	\begin{lemma}\label{lm:spset}
		For even $N \ge 2$, we can compute a compatible $(M,N)$-ISS pair $(\aset,\bset)$ with $M \ge 2^{N/2-1}$ in time polynomial in $M$ and $N$. Further, $\bset = \bar{\aset}$.
	\end{lemma}
	\begin{proof}
		Given even $N \ge 2$, we construct two set systems $\aset =(U,S_A)$ and $\bset=(U,S_B)$ greedily as follows. First set $U = [N]$, and unmark all the subsets $s \subseteq U$ of size $\nicefrac{N}{2}$. Then, for every subset $s \subseteq U, |s|=\nicefrac{N}{2}$ that is unmarked, add $s$ to $S_A$ and $\bar{s}$ to $S_B$, and mark both $s$ and $\bar{s}$.
		Note that $|S_A|= |S_B|$, and $\bset = \bar{\aset}$. First, we claim that both $\aset$ and $\bset$ are $(|S_A|, N)$-ISS. Indeed, observe that any $s, t \in S_A$  $(s,t\in S_B)$ intersects since $|s| = |t| = N/2$ and $t \ne \bar{s}$. Next, for lower bounding $M$, we have $|S_A| = |S_B| = \frac{1}{2} {N \choose {N/2}} \ge 2^{N/2-1}$.
		The total time to construct $(\aset,\bset)$ is $2^{N} \poly(N) = \poly(M,N)$.
        Finally, to see that $(\aset,\bset)$ is a compatible ISS pair, consider the bijection $f: S_A \mapsto S_B$ with $f(s) = \bar{s}$, for $s \in S_A$, and note that for $s \in S_A$,  $\aset_s = (U,(S_A \setminus \{s\}) \cup \{f(s)\})$ is an ISS. This is because for $s \in S_A$ and any ${t} \in {S_A} \setminus \{s\}$, since $|\bar{s}| = |\bar{t}| = \frac{N}{2}$ and $\bar{t} \ne \bar{s}$, we have $\bar{s} \setminus \bar{t} \ne \emptyset$. But $\bar{s} \cap {t} = \bar{s} \setminus \bar{t} \ne \emptyset$. Hence, $f(s)$ intersects $t$.
		\qed
	 \end{proof}

	\subsection{Hardness of \cpsp} \label{ss:dihardness}
	Our hardness result follows from the following \fpt-reduction from \sgi that yields compact instances of \psp and the fact that \sgi is \wonehard.
	\begin{theorem} \label{thm:sgired}
		There is an \fpt-reduction that, for every instance $\iset = (G=(V_G,E_G),H=(V_H,E_H),k)$ of \sgi with $|V_G|=n$ and $|E_G|=m$, computes $\mu = O(k!)$ instances $\jset_p=(\uset_p,\ssset_p,r), p \in [\mu]$ of  \psp with  $|\uset_p|= \Theta(k \log n)$, $|\ssset_p|= \Theta(n^2k + mk)$, and $r=\Theta(k)$,
		such that there is a subgraph of $G$ isomorphic to $H$ if and only if there is an $r$-packing in at least one of the instances $\{\jset_p\}_{p \in [\mu]}$.
	\end{theorem}
	
	\begin{proof}
	    The construction follows the approach outlined in Section~\ref{ss:tech}.
		Let $V_G = \{1,\cdots,n\}$. Let $(\aset,\bar{\aset})$ be the compatible $(M,N)$-ISS pair given by Lemma \ref{lm:spset}, for $N = 2\lceil \log (n+1) \rceil +2$. We call $\aset= (U_A,S_A)$  and $\bar{\aset}= (U_A,\cmp(S_A))$ as the base ISS gadgets.
		Further, assume an arbitrary ordering on $S_A= \{s^1,\cdots,s^M\}$. 
		Since $M \ge 2^{N/2-1}  > n$, every $\alpha \in V_G$ can  be identified by the set $s^\alpha \in S_A$ corresponding to the index $\alpha \in [M]$. 
		For each ordering $p: V_H \rightarrow [\ell]$, create an instance $\jset_p = (\uset_p, \ssset_p,r)$ of \cpsp as follows. Rename the vertices of $V_H$ as $\{v_1,\cdots,v_\ell \}$ with $v_i := v \in V_H$ such that $p(v)=i$. 
		For each $v_i \in V_H$, create a collection $\cset_{v_i}$ of $d(v_i)+1$ many different copies of base ISS gadget $\aset$ (i.e., each has its own distinct universe) as:
		$
		\cset_{v_i} :=\{\aset_{v_i,0}, \{\aset_{v_i,w}\}_{w \in N(v_i)} \}
		$,
		where $\aset_{v_i,0} = (U_{v_i,0},S_{v_i,0}) \text{ and } \aset_{v_i,w} = (U_{v_i,w},S_{v_i,w}).$
	    Let  $U_{v_i} = \cup_{w \in N(v)} U_{v_i,w}$.
		For each $\cset_{v_i}$, let $U_{\cset_i} = U_{v_i,0} \cup U_{v_i}$. Now, we define the universe $\uset_p$ of $\jset_p$ as $	\uset_p = \bigcup_{i \in [\ell]} U_{\cset_i}$.
		
        The sets in $\ssset_p$ are of two types: $V$-sets and $E$-sets as defined below.
		For $\alpha \in V_G$ and $v \in V_H$, denote by $S^\alpha_v = \cup_{w \in N(v)} s^\alpha_{v,w}$. Recall that for $\alpha \in V_G$ and $(v,w) \in E_H$, the set $s^\alpha_{v,w}$ is the $\alpha^{th}$ set in $S_{v,w}$ of ISS $\aset_{v,w}=(U_{v,w},S_{v,w})$. 
		
		\textit{$V$-sets:} For each $\alpha \in V_G$, for each $v_i \in \{v_1,\cdots, v_{\ell-1}\}$, and for each $\beta \in V_G, \beta > \alpha$, add a set $S_{\alpha \mapsto v_i, \beta}$ to $\ssset_p$ such that 
		\[
		S_{\alpha \mapsto v_i, \beta} := {s}^{\alpha}_{v_i,0} \; \bigcup \; S^\alpha_{v_i} \;  \bigcup \; \bar{s}^{\beta}_{v_{i+1},0}
		\]
		Further, for  each $\alpha \in V_G$, and for each $\beta \in V_G, \beta < \alpha$, add a set $S_{\alpha \mapsto v_\ell, \beta}$ to $\ssset_p$ such that 
		\[
		S_{\alpha \mapsto v_\ell, \beta} := {s}^{\alpha}_{v_\ell,0} \; \bigcup \; S^\alpha_{v_\ell} \;  \bigcup \; \bar{s}^{\beta}_{v_{1},0} 
		\]
		
		\textit{$E$-sets:} For each edge $(\alpha,\beta) \in E_G$ and each edge $(v_i,v_j) \in E_H$, add a set $S_{(\alpha,\beta) \mapsto (v_i,v_j)}$ to $\ssset_p$ such that
		\[
		S_{(\alpha,\beta) \mapsto (v_i,v_j)} := \bar{s}^\alpha_{v_i,v_j}  \bigcup \bar{s}^{\beta}_{v_j,v_i}
		\]
		
		\textit{Parameter:} Set $r := k + \ell$.

   This concludes the construction. Before we prove its correctness, we note the size of the constructed instance $\jset_p$. First,  $r= \Theta(k)$, since $\ell \le 2k$. Then, 
		$
		|\uset_p| = \sum_{i=1}^{\ell} |U_{\cset_i}| = \sum_{i=1}^{\ell} (d(v_i) + 1) N = \Theta(k\log n),
		$
		and
		$|\ssset_p| = \Theta( n^2 \ell + mk) = \Theta(n^2k)$.
  
        \paragraph{\textbf{Yes case.}}
        Suppose there is a subgraph $G'=(V_{G'},E_{G'})$ of $G$ that is isomorphic to $H$ with injection $\phi:V_{H} \rightarrow V_{G'}$. Let  $V_{G'}=\{\alpha_{1},\alpha_{2},\cdots, \alpha_{\ell}\} \subseteq [n]$ such that $\alpha_1<\alpha_{2} <\cdots<\alpha_{\ell}$. 
		Relabel the vertices of $H$ as $\{v_1,\cdots,v_\ell\}$, where $v_i := \phi^{-1}(\alpha_{i}), i \in [\ell]$.
		Now, consider the ordering $p$ of $V_H$ such that $p(v_i)=i$, for $i \in [\ell]$, and fix the corresponding instance $\jset_p=(\uset_p,\ssset_p,r)$. Consider the following collection of $V$-sets and $E$-sets:
		$
		T_V := \bigcup_{i \in [\ell]} S_{\alpha_{i} \mapsto v_i, \alpha_{i+1}}$ and $T_E := \bigcup_{(v_i,v_{j}) \in E_H} S_{(\alpha_{i}, \alpha_{j}) \mapsto (v_i,v_{j})}.
		$
		Let $T=T_V \cup T_E$. Note that for this choice of $p$, we have $T_V \subseteq \ssset_p$ due to construction, and $T_E \subseteq \ssset_p$ due to $\phi$, and hence $T \subseteq \ssset_p$. Further, $|T| = |T_V| + |T_E| = \ell + k=r$, as required. Now,  we claim that $T$ forms a packing in $\jset_p$. 
		Towards this goal, note that it is sufficient to show that the sets in $T\mid_{U_{\cset_i}}$ are mutually disjoint, for all $i \in [\ell]$. To this end, it is sufficient to show that both $T \mid_{U_{v_i,0}}$ and $T \mid_{U_{v_i}}$  are packing, for all $i \in [\ell]$. Fix $i \in [\ell]$, and consider the following cases:
		\begin{enumerate}
			
        \item  $T \mid_{U_{v_i,0}}:$ Since $T_E \mid_{U_{v_i,0}} = \emptyset$ by construction, we focus on $T_V \mid_{U_{v_i,0}}$.
			But, $T_V \mid_{U_{v_i,0}}$ is a packing since,
			\begin{align*}
			T_V \mid_{U_{v_i,0}} =
			\begin{cases}
			\{S_{\alpha_{i-1} \mapsto {v_{i-1}}, \alpha_i}\mid_{U_{v_i,0}},  S_{\alpha_{i} \mapsto v_i, \alpha_{i+1}}\mid_{U_{v_i,0}}\} = \{\bar{s}^{\alpha_i}_{v_i,0}, s^{\alpha_i}_{v_i,0}\}, \text{if }i \ne 1 \\
			\{S_{\alpha_{\ell} \mapsto {v_{\ell}}, \alpha_1}\mid_{U_{v_1,0}},  S_{\alpha_{1} \mapsto v_1, \alpha_{2}}\mid_{U_{v_1,0}}\}=\{\bar{s}^{\alpha_1}_{v_1,0}, s^{\alpha_1}_{v_1,0}\}, \text{ if }i = 1. 
			\end{cases}
			\end{align*}
			
		  \item $T\mid_{U_{v_i}}:$ It is sufficient to show that $T \mid_{U_{v_i,v_j}}$ is a packing, $\forall v_j \in N(v_i)$. But this follows since, $\forall v_j \in N(v_i)$,
			\begin{align*}
			T \mid_{U_{v_i,v_j}} &= \{T_V \mid_{U_{v_i,v_j}}, T_E \mid_{U_{v_i,v_j}}\} \\
			&= \{ S_{\alpha_i\mapsto v_i, \alpha_{i+1}}\mid_{U_{v_i,v_j}}, S_{(\alpha_i,\alpha_j) \mapsto (v_i,v_j)}\mid_{U_{v_i,v_j}} \}
			= \{s^{\alpha_i}_{v_i,v_j},\bar{s}^{\alpha_i}_{v_i,v_j} \}.
			\end{align*}
		\end{enumerate}
  
		\paragraph{\textbf{No case.}}
		Suppose there is an $r$-packing $T \subseteq \ssset_p$ in some instance $\jset_p, p\in[\mu]$, then we show that there is a subgraph $G_T$ of $G$ that is isomorphic to $H$. First note that $p \in [\mu]$ gives a labeling $\{v_1,\cdots, v_\ell\}$ of $V_H$ such that $v_i=p^{-1}(i)$, for $i\in [\ell]$.
		Next, partition $T$ into $T_V$ and $T_E$, such that $T_V$ and $T_E$ correspond to the $V$-sets and $E$-sets of $T$ respectively. This can be easily done since $t \in T$ is a $V$-set if and only if $t\mid_{U_{v_i,0}} = s^{\alpha}_{v_i,0} \in \aset_{v_i,0}$, for some $\alpha \in V_G, v_i \in V_H$. Let
		$
		U_0 = \{U_{v_i,0}\}_{v_i \in V_H}$ {and} $U_1 = \{U_{v_i}\}_{v_i \in V_H}
		$.
		We claim the following.
		\begin{lemma}\label{lm:setsizes}
			$|T_V| = \ell$ and $|T_E| = k$.  
		\end{lemma}
		\begin{proof}
			Note that for $t \in T_V$, we have $t\mid_{U_0} = \{s^{\alpha}_{v_i,0}, \bar{s}^{\beta}_{v_{i+1},0}\}$, for some $\alpha,\beta \in V_G$ and $v_i, v_{i+1} \in V_H$. 
			Hence, it follows that $|t\mid_{U_0}| = N$. Since $|U_0| = \ell N$ and $T_V$ is a packing, we have $|T_V| \le \ell$. 
			For bounding $|T_E|$, consider  $t\in T_E$, and note that $t \mid_{U_1} = \{\bar{s}^\alpha_{v_i,v_j},\bar{s}^{\beta}_{v_j,v_i}\}$, for some $(\alpha,\beta) \in E_G$ and $(v_i,v_j) \in E_H$. 
			But also note that we have $\bar{s}^\alpha_{v_i,v_j} \in \bar{\aset}_{v_i,v_j}$ and $\bar{s}^\beta_{v_j,v_i} \in \bar{\aset}_{v_j,v_i}$.
			Hence, by the virtue of $T_E$ being a packing and using the facts that $U_1$ is the union of universes of $2k$ many base ISS $\{\bar{\aset}_{v,w}\}_{v \in V_H, w \in N(v)}$, and each $t \in T_E$ contains sets from two of such ISS, it follows $|T_E| \le k$.
			Finally, $|T| = r = \ell + k$ implies $|T_V| = \ell$ and $|T_E| = k$.   \qed
		\end{proof}
        For $i \in [\ell]$, as $\aset_{v_i,0}=(U_{v_i,0}, S_{v_i,0})$ is an ISS, we can relabel the sets in $T_V$ as $T_V = \{T^1_{V},\cdots, T^\ell_{V}\}$, where $T^i_V := t \in T_V$ such that $t\mid_{U_0} \ni s^\alpha_{v_i,0}$, for some $s^\alpha_{v_i,0} \in S_{v_i,0}$. 
  The following lemma is our key ingredient.
		
		\begin{lemma} \label{cl:tcovu}
			$T$ covers the whole universe $\uset_p$. 
		\end{lemma}
			\begin{proof}
			Since $\uset_p = U_0 \cup U_1$, we will show that $T\mid_{U_j}$  covers $U_j$, for $j=\{0,1\}$. For $U_0$, note that $T\mid_{U_0} = T_V\mid_{U_0}$ by construction. For $T^i_V\in T_V$, we have $|T^i_V\mid_{U_0}| = N$ due to complement partition axiom of $(\aset_{v_i,0},\bar{\aset}_{v_i,0})$. Since  $T_V$ forms a packing, we have that $|\cup_{i \in [\ell]} T^i_V \mid_{U_0}| = \ell N = |U_0|$, as desired.
			Next, we have $|U_1| = 2kN$. Consider $T^i_V\in T_V$ and notice $|T^i_V\mid_{U_1}|=\frac{N}{2}d(v_i)$ since $T^i_V \mid_{U_1} = S^\alpha_{v_i}$, for some $\alpha \in V_G$. Since $T_V$ forms a packing, we have 
			$
			|\bigcup_{i \in [\ell]} T^i_V \mid_{U_1}|= \sum_{i=1}^{\ell} |T^i_V\mid_{U_1}|  = kN
			$.
			Now consider $t = S_{(\alpha,\beta) \mapsto (v_i,v_j)} \in T_E$, for some $(\alpha,\beta) \in E_G$ and $(v_i,v_j) \in E_H$. Since, $t\mid_{U_1} =  \{\bar{s}^\alpha_{v_i,v_j}, \bar{s}^{\beta}_{v_j,v_i}\}$, we have $|t\mid_{U_1}| = N$. As $T_E$ forms a packing, we have
			$
			|\bigcup_{t \in T_E} t \mid_{U_1}| = \sum_{t \in T_E} |t\mid_{U_1}| =  kN
			$.
			Finally,  $T$ being  a packing, we have
			$
			|\bigcup_{\tau \in T} \tau  \mid_{U_1}| = |\bigcup_{i \in [\ell]} T^i_V \mid_{U_1}| + |\bigcup_{t \in T_E} t \mid_{U_1}| = 2kN = |U_1|
			$
			as desired.
			\qed
		\end{proof}
		Let $\alpha_i = \alpha \in V_G$ such that $T^i_V\mid _{U_{v_i,0}} \ni s^{\alpha}_{v_i,0}$, for $i \in [\ell]$. Let $V_T = \{\alpha_i\}_{i \in [\ell]} $.
        The following lemma asserts that $|V_T|=\ell$. 
		\begin{lemma}\label{lm:ethuvertex}
			For each vertex $\alpha \in V_G$, there is at most one $V$-set $S_{\alpha \mapsto v_i, \beta}$ in $T_V$, for some $v_i \in V_H$ and $\beta \in V_G$.
		\end{lemma}
		\begin{proof}
		    It is sufficient to show  $\alpha_{i} < \alpha_{i+1}$, for $i \in [\ell-1]$. Fix such $i$ and consider
			the universe $U_{v_i,0}$ of  $\aset_{v_{i+1},0}$. Then, note that only $T^i_V$ and $T^{i+1}_V$ contain elements of $U_{v_i,0}$. Let $T^i_V = S_{\alpha_i \mapsto v_i, \beta}$ for $\beta > \alpha_i$, and let $T^{i+1}_V =  S_{\alpha_{i+1} \mapsto v_{i+1}, \gamma}$, for $\gamma > \alpha_{i+1}$. As $T$ covers $U_{v_i,0}$ (Lemma~\ref{cl:tcovu}),  and using the complement partition property of the compatible ISS pair  $(\aset_{v_{i+1},0},\bar{\aset}_{v_{i+1},0})$, we have that $\alpha_{i+1} = \beta > \alpha_i$.\qed
		\end{proof}
		\begin{lemma} \label{lm:ethuedge}
			For every edge $(\alpha,\beta) \in E_G$, there is at most one $E$-set $S_{(\alpha,\beta) \mapsto (v_i,v_j)}$ in $T_E$, for some $(v_i,v_j) \in E_H$.
		\end{lemma}
  		\begin{proof}
			Suppose there are two sets $S_{(\alpha,\beta) \mapsto (v_i,v_j)}, S_{(\alpha,\beta) \mapsto (v'_i,v'_j)} \in T_E$, for some $(\alpha,\beta) \in E_G$. Without loss of generality assume $v_i \ne v'_i$. Then, we will show that  $S_{\alpha \mapsto v_i, \gamma}, S_{\alpha \mapsto v'_i, \delta} \in T_V$, for some $\gamma,\delta \in V_G$, contradicting Lemma~\ref{lm:ethuvertex}. Since $S_{(\alpha,\beta) \mapsto (v_i,v_j)}, S_{(\alpha,\beta) \mapsto (v'_i,v'_j)} \in T_E$, it holds that $T_E\mid_{U_{v_i,v_j}} = \bar{s}^\alpha_{v_i,v_j}$, and $T_E\mid_{U_{v'_i,v'_j}}=\bar{s}^\alpha_{v'_i,v'_j}$.  As $T$ covers $\uset_p$, in particular, $T$ covers $U_{v_i,v_j}$, it must be that $T_V\mid_{U_{v_i,v_j}} =  s^\alpha_{v_i,v_j}$ as $(\aset_{v_i,v_j}, \bar{\aset}_{v_i,v_j})$ is a compatible ISS pair. By similar reasoning for $U_{v'_i,v'_j}$, it must be that $T_V\mid_{U_{v'_i,v'_j}} =  s^\alpha_{v'_i,v'_j}$.
			This implies that $T_V \mid_{U_{v_i}} = S^\alpha_{v_i}$ and  $T_V \mid_{U_{v'_i}} = S^\alpha_{v'_i}$. Thus, $S_{\alpha \mapsto v_i,\gamma}, S_{\alpha \mapsto v'_i,\delta} \in T_V$ for $v_i \ne v_i'$, for some $\gamma,\delta \in V_G$.\qed
		\end{proof}
Let $G_T = G[V_T] =(V_T,E_T)$, be the induced subgraph of $G$ on $V_T$.  
To finish the proof, we claim  that $G_T$ is isomorphic to $H$ with the injective homomorphism $\phi : V_H \rightarrow V_{T}$ given by $\phi(v_i) = \alpha_{i}$, for $i \in [\ell]$.
To this end, we will show that for any $(v_i,v_j) \in E_H$, it holds that $(\phi(v_i),\phi(v_j))= (\alpha_i,\alpha_j) \in E_T$. 
Consider the universe $U_{v_i,v_j}$, and note that $T^i_V \mid_{U_{v_i,v_j}}= s^{\alpha_i}_{v_i,v_j}$.
As $T$ covers $U_{v_i,v_j}$, it holds that $T_E \mid_{U_{v_i,v_j}} = \bar{s}^{\alpha_i}_{v_i,v_j}$ since $(\aset_{v_i,v_j}, \bar{\aset}_{v_i,v_j})$ is a compatible ISS pair. Hence $S_{(\alpha_i,\beta) \mapsto (v_i,v_j)} \in T_E$, for some $(\alpha_i,\beta) \in E_G$. This implies that $T_E \mid_{U_{v_j,v_i}} = \bar{s}^{\beta}_{v_j,v_i}$. By similar arguments for $U_{v_j,v_i}$, we have that $\beta = \alpha_j$ as $T^j_v \mid_{U_{v_j,v_i}}= s^{\alpha_j}_{v_j,v_i}$. Hence $(\alpha_i,\alpha_j)= (\phi(v_i),\phi(v_j)) \in E_G$.
\qed
	\end{proof}
 	\textbf{Acknowledgments.} 
    This work has been partially supported by European Research Council (ERC) under the European Union’s Horizon 2020 research and innovation programme (grant agreement No. 759557).
  I thank Parinya Chalermsook for the informative discussions about the results of the paper, and for providing guidance on writing this paper. I also thank anonymous reviewers for their valuable suggestions on improving the readability of the paper.
	
	

    	\appendix

	\section{A simple algorithm for \psp}\label{ss:dialgo}
		In this section we show an algorithm for \spp that runs in  time $O^*(2^{|\uset|})$.
		The main idea is to exploit the fact that $|\uset|$ is small and convert the problem to the path finding problem in directed acyclic graphs (DAG). To this end, we enumerate all the subsets of $\uset$ by creating a vertex for each subset. Then, we add a directed edge from subset $T_1$ to subset $T_2$ if there is a set $S_i \in \ssset$ such that $T_1 \cap S_i = \emptyset$ and $T_1 \cup S_i = T_2$. Intuitively, the edge $(T_1,T_2)$ in the DAG captures the fact that, if the union of our present solution is $T_1$, then we can improve it by including $S_i$ to get a solution whose union is $T_2$. Thus finding a maximum sized packing reduces to finding a longest path in the DAG which can be found efficiently by standard dynamic programming technique.

	\begin{theorem} \label{thm:algosp}
		There is an algorithm for \spp  running in time $O^*(2^{|\uset|})$.
	\end{theorem}
	\begin{proof}
		Given an instance $(\uset,\ssset)$ of \spp, the idea is to construct a graph $G=(V,E)$ such that there is a vertex in $G$ for every subset of $\uset$. For a subset $T \subseteq \uset$, let $v_T$ be the corresponding vertex in $G$.  Then, add a labeled directed edge from vertex $v_{T_i}$ to $v_{T_j}$ with label $S_k \in \ssset$ if there is $S_k \in \ssset$ such that ${T_i} \cup S_k = {T_j}$ and ${T_i} \cap S_k = \emptyset$. In other words, adding the set $S_k$ to our present solution, whose union is denoted by $T_i$,  is safe and results into a new solution whose union is $T_j$. \ignore{Finally, label the directed edge $(v_{S_i}, v_{S_j})$ by the set $T$.} First note that $G$ can be constructed from $(\uset,\ssset)$ in time   $O(2^{|\uset|}|\ssset||\uset|)$.  Now we claim that there are $\ell$ pairwise disjoint sets in $(\uset,\ssset)$ if and only if there is a path of length $\ell$ starting at vertex $v_\emptyset$ in $G$.
		
		For one direction, suppose there is a path $P=\{v_\emptyset, v_{T_1},v_{T_2},\cdots,v_{T_\ell}\}$ of length $\ell$ in $G$ . Then, there are $\ell$ sets in $\ssset$ labeled by the edges of $P$ that are pairwise disjoint by construction. For the other direction, let $S=\{S_1,S_2,\cdots, S_\ell\}$ be  pairwise disjoint sets in $\ssset$. Now, fix some order on the sets of $S$ and  consider the following sets, for $i \in [\ell]$,	
		$
		T_i := \bigcup_{j=1}^i S_j
		$.	
		Now consider the collection of vertices $P=\{v_{T_0}, v_{T_1},\cdots,v_{T_\ell}\}$, where ${T_0} := \emptyset$. It is easy to see that $P$ is a path in $G$ since  there is an edge from $v_{T_i}$ to $v_{T_{i+1}}$ in $G$, labeled by set $S_{i+1} \in \ssset$, for every $i \in \{0,1,\cdots,\ell-1\}$. 
		
		Next we show how to find an $\ell$ length path in $G$ starting $v_\emptyset$ in time $O^*(2^{|\uset|})$. First note that $G$ is a directed acyclic graph(DAG) since every directed edge $(v_{T_i}, v_{T_j})$ in $G$ implies $|T_j| > |T_i|$. Now we  use standard dynamic program to find an $\ell$ length path at $v_\emptyset$. Rename the vertices of $G$ to $\{0,1,\cdots, N-1\}$, where $N=2^{|\uset|}$, such that the vertex $v_{T_j}$ is renamed to $i \in \{0,1,\cdots, N-1\}$ where $i$ is the number whose binary representation corresponds to the characteristic vector $\chi_{T_j}$. Next, define a two dimensional bit array $\mathbf{A}$ such that for $0\le i \le N-1$ and $1 \le j \le \ell$,
		\begin{align*}
		\mathbf{A}[i,j] := 
		\begin{cases}
		1 & \text{if there is a path from vertex $0$ to vertex $i$ of length $j$}   \\
		0 & \text{otherwise}
		\end{cases}
		\end{align*}
		Computing bottom up, we can fill $\mathbf{A}$ in time $O(\ell (N+M))$, where $M := 2^{|\uset|}|\ssset|$ is the number of edges in $G$, and then look  for entry $1$ in the array $\mathbf{A}[\cdot, \ell]$. Further, to find a path, we can modify the algorithm such that, instead of storing a bit, $\mathbf{A}[i,j]$ now stores one of the paths (or just the preceding vertex of a path). Thus we can find an optimal solution of \spp in time $O^*(2^{|\uset|})$.\qed 
	\end{proof}

	\section{Hardness of \crvs} \label{sec:ckvs}
	In this section, we give a proof sketch of Theorem~\ref{thm:ckvs}. First, we define the problem.
	\begin{definition}[\rvsm]
		Given a collection $\mathcal{C} = \{\vec{\eta_1},\cdots, \vec{\eta_N}\}$ of $d$ dimensional vectors over $\mathbb{F}_2$, a vector $\vec{b} \in \mathbb{F}_2^d$, and an integer $k$, the \rvsm problem asks if there is $I \subseteq  [N], |I|= k$ such that $\sum_{{j} \in I} \vec{\eta_{j}}=\vec{b}$, where the sum is over $\mathbb{F}_2^d$.
		An instance of $\rvsm$ is denoted as $(\cset,\vec{b},d,r)$.	
	\end{definition}
    \crvs is defined when $d = f(r) \cdot \Theta(\poly(\log N))$, for some $f(r)\ge r$.	
	The hardness  of \crvs follows from the following theorem.
	\begin{theorem} \label{thm:sgired_rvsm}
		There is an \fpt-reduction that, for every instance $\iset = (G=(V_G,E_G),H=(V_H,E_H),k)$ of \sgi with $|V_G|=n$ and $|E_G|=m$, computes $\mu = O(k!)$ instances $\lset_p=(\cset_p,\vec{b},d,r), p \in [\mu]$, of  \rvsm with the following properties:
		\begin{itemize}
			\item [$\circ$] $d= \Theta(k \log n)$
			\item[$\circ$] $|\cset_p|= \Theta(n^2k + mk)$
			\item[$\circ$] $r=\Theta(k)$
		\end{itemize}
		such that there is a subgraph of $G$ isomorphic to $H$ if and only if there exists $p \in [\mu]$ such that there are at most $r$ vectors in the instance $\jset_p$ that sum to $\vec{b}$.
	\end{theorem}
	\paragraph*{Proof Sketch.} 
	The first part of the reduction is, in fact, same as that described in Theorem~\ref{thm:sgired}, with a simple observation that any optimal packing in the instance generated by Theorem~\ref{thm:sgired} is also a covering. This holds true in No case due to Lemma~\ref{cl:tcovu}, and it holds true in Yes case due to the complement exchange property of the compatible ISS-pair gadget used in the construction.
 We call such solution as \textit{exact cover}. In the second part, we transform this instance of Theorem~\ref{thm:sgired} to an instance of \rvsm. The following definition is useful for the transformation.
	\begin{definition}[Characteristic vector]
		Let $U$ be a universe of $q$ elements. Fix an order on the elements of $U = (e_1,\cdots,e_q)$. For any set $S \subseteq U$, define the characteristic vector $\vec{\chi}_S \in \mathbb{F}_2^q$ of $S$ as follows. The $t^{th}$ co-ordinate of $\vec{\chi}_S$,
		\[
		\vec{\chi}_S(t) :=
		\begin{cases}
		1 & \text{ if } e_t \in S,\\
		0 & \text{ if } e_t \notin S.
		\end{cases}
		\]
	\end{definition}
	For every instance $\jset_p=(\uset_p,\ssset_p,r), p \in [\mu]$ generated by Theorem~\ref{thm:sgired}, we create an instance $\lset_p=(\cset_p,\vec{b},d,r)$ of \rvsm as follows. Rename the vertices of $V_H$ as $\{v_1,\cdots,v_\ell \}$ such that $v_i := v \in V_H$ such that $p(v)=i$. Note that this induces an ordering on $V_H$ as $v_1 < \cdots < v_\ell$. Thus, for $v_i \in V_H$, we have an ordering on $N(v_i) = \{v'_1,\cdots,v'_{d(v_i)}\}$ as $v'_1 < \cdots < v'_{d(v_i)}$. Hence, for $\lambda \in [d(v_i)]$, we call $v'_\gamma$ as the $\lambda^{th}$ neighbour of $v_i$. Now, for $v_i \in V_H$, we define $\Gamma_i : N(v_i) \mapsto [d(v_i)]$ as $\Gamma_i(v_j) := \lambda \in [d(v_i)]$, such that $v_j$ is the $\lambda^{th}$ neighbour of $v_i$. Next, we construct vectors corresponding to the sets in $\ssset_p$.
	For every $V$-set $S_{\alpha \mapsto v_i}$, for  $v_i \in V_H$ and $\alpha \in V_G$, of $\ssset_p$, define $|\uset_p| + \ell + 2k$ length vector $\vec{\chi}'_{S_{\alpha \mapsto v_i}}$ as follows.
	\[
	\vec{\chi}'_{S_{\alpha \mapsto v_i}}(t)  :=
	\begin{cases}
	\vec{\chi}_{S_{\alpha \mapsto v_i}}(t) & \text{ if } t \in [|\uset_p|],\\
	1 & \text{ if } t = |\uset_p| + i, \\
	0 & \text{ otherwise.}
	\end{cases}
	\]
	Similarly, for every $E$-set $S_{(\alpha,\beta) \mapsto (v_i,v_j)}$, for  $(\alpha,\beta) \in E_G$ and $(v_i,v_j) \in E_H$, define  $|\uset_p| + \ell + 2k$ length vector $\vec{\chi}'_{S_{(\alpha,\beta) \mapsto (v_i,v_j)}}$ as follows.
	\[
	\vec{\chi}'_{S_{(\alpha,\beta) \mapsto (v_i,v_j)}}(t)  :=
	\begin{cases}
	\vec{\chi}_{S_{(\alpha,\beta) \mapsto (v_i,v_j)}}(t) & \text{ if } t \in [|\uset_p|],\\
	1 & \text{ if } t = |\uset_p| + \ell + \sum_{\rho=1}^{i-1} d(v_\rho) + \Gamma_i(j), \\
	1 & \text{ if } t = |\uset_p| + \ell + \sum_{\rho=1}^{j-1} d(v_\rho) + \Gamma_j(i), \\
	0 & \text{ otherwise.}
	\end{cases}
	\]    
	Now consider the instance $\lset_p=(\cset_p,\vec{b},d,r)$, where
	\begin{itemize}
		\item $d := |\uset_p| + \ell + 2k$
		\item $\cset_p := \{\vec{\chi'}_S\}_{S \in \ssset_p}$
		\item $\vec{b} := \vec{1}$, the all ones vector.
	\end{itemize}
	Before we prove the correctness, we define some notations. The first $\ell$ bits that we appended to $\vec{\chi}_S$ are called $V$-indicator bits, and the next $2k$ bits are called $E$-indicator bits. The universe corresponding to $V$-indicator bits and $E$-indicator bits is denoted as $U_2$ and $U_3$ respectively. \\
	\textbf{\textsf{Yes Case: }} From Theorem~\ref{thm:sgired}, there is $p \in [\mu]$ such that $\jset_p=(\uset_p,\ssset_p,r)$ has a $r$-packing $\ssset'_p \subseteq \ssset_p$ that covers $\uset_p$. Then, consider the corresponding instance $\lset_p=(\cset_p,\vec{b},d,r)$ of \rvsm. Then, note that
	\[
	\sum_{S \in \ssset'_p} \vec{\chi_S} = \vec{1} = \vec{b}
	\]
	since $\ssset'_p$ is a packing covering $\uset_p$. Hence $\{\vec{\chi}_S\}_{S \in \ssset'_p} \subseteq \cset_p$ is a solution to $\lset_p$.\\
	\textbf{\textsf{No Case: }} Let $W \subseteq{\mathcal{C}_p}, |W|=r$, be a solution of $\lset_p$, for some $p \in [\mu]$. Let $T \subseteq \ssset_p$ be the corresponding collection of sets in $\jset_p$ to $W$. Let $T_V$ and $T_E$ be the $V$-sets and $E$-sets of $T$ respectively. 
	Let $W_V$ and $W_E$ be the set of vectors corresponding to $T_V$ and $T_E$ respectively. We say vectors in $W_V$  and $W_E$ as $V$-vectors and $E$-vectors respectively. Note that $W= W_V \dot\cup W_E$ as $T=T_V \dot\cup T_E$. The following lemma is equivalent to Lemma~\ref{lm:setsizes}.
	\begin{lemma}\label{lm:vecsizes}
		$|W_V| = \ell$ and $|W_E| = k$.     
	\end{lemma}
	\begin{proof}
		First consider $W_V$, and note that only $V$-vectors have $V$-indicator bits set to $1$. Since a $V$-vector has at most one $V$-indicator bit set to $1$, and there are $\ell$ $V$-indicator bits set to $1$ in $\vec{b}$, it follows that $|W_V| \ge \ell$. Now consider $W_E$, and note that only $E$-vectors have $E$-indicator bits set to $1$. Since a $E$-vector has at most two $E$-indicator bit set to $1$, and there are $2k$ $E$-indicator bits set to $1$ in $\vec{b}$, it follows that $|W_E| \ge k$. Since, $|W| = \ell + k$, we have $|W_V| =\ell$ and $|W_E|=k$.\qed
	\end{proof}
	Now we claim that $T$ is an exact $r$-cover in $\jset_p$. It is sufficient to show $T$ is an $r$-packing since $T$ covers $\uset_p$.
	To this end, note that Lemma~\ref{lm:vecsizes} implies that $|T_V|=\ell$ and $|T_E| = k$. This implies that that $T\mid_{U_2}$ and $T\mid_{U_3}$ is a packing. Hence, $T\mid_{U_2 \cup U_3}$ is a packing. Let $U'_0 := U_0 \cup U_2 \cup U_3$ and $U'_1 := U_1 \cup U_2 \cup U_3$.
	Next consider $T_V$ and note that each set in $T_V$ covers exactly $N$ elements of $U_0$. Thus, the total number of elements covered by $T_V$ is at most $\ell N$. But then, since $T_V$ covers $U_0$ and $|U_0| = \ell N$, it follows that the every set in $T_V$ must cover different elements of $U_0$. Hence, $T_V\mid_{U'_0}$ is a packing. On the other hand, $T_V$ covers at most $kN$ elements of $U_1$ as each set in $T_V$ covers exactly $\frac{N}{2} d(v_i)$ elements of $U_1$, for some $v_i \in V_H$. Since, $T$ covers $U_1$, it must be that $T_E$ must cover at least $kN$ elements of $U_1$, as $|U_1|=2kN$. Since each set in $T_E$ covers $N$ elements of $U_1$, from Lemma~\ref{lm:vecsizes} it follows that $T_E$ covers at most $kN$ elements of $U_1$. Thus, the sets in $T_E$ must cover different elements of $U_1$, and hence $T_E\mid_{U'_1}$ is a packing. However, $T_E\mid_{U'_1} = T_E$ since sets in $T_E$ only contain elements of $U'_1$. Hence, $T_E$ is a packing. This means that $T_V$ must cover at least $kN$ elements of $U_1$. Then, it follows that $T_V\mid_{U'_1} \cup T_E\mid_{U'_1} = T\mid_{U'_1}$ is a packing since we observed above that $T_V$ covers at most $kN$ elements of $U_1$. Thus, $T$ is a packing as $T\mid_{U_0} = T_V\mid_{U_0} \cup T_E\mid_{U_0} = T_V\mid_{U_0}$ is also a packing.
	
	Now since $T$ is an exact $r$-cover of $\uset_p$, we can use the No case of Theorem~\ref{thm:sgired} to recover the isomorphic subgraph of $G$ to $H$, which finishes the proof.	  
	\qed
	
	

\end{document}